\documentclass[11pt]{article}
\usepackage{amsmath}
\usepackage{amssymb}
\usepackage{amsthm}
\usepackage{url}
\usepackage{epic,gastex}

\newtheorem{problem}{Problem}
\newtheorem{theorem}{Theorem}
\newtheorem{lemma}[theorem]{Lemma}
\newtheorem{corollary}[theorem]{Corollary}

\theoremstyle{remark}
\newtheorem{example}{Example}

\newcommand{\sa}{synchronizing automata}
\newcommand{\san}{synchronizing automaton}
\newcommand{\sw}{reset word}
\newcommand{\sws}{reset words}

\DeclareSymbolFont{rsfscript}{OMS}{rsfs}{m}{n}
\DeclareSymbolFontAlphabet{\mathrsfs}{rsfscript}

\newcommand{\bigand}{\operatornamewithlimits{\hbox{\Large$\&$}}}

\begin{document}
\title{P(l)aying for Synchronization\thanks{Supported by the Russian Foundation for Basic Research, grant 10-01-00793,
and by the Presidential Program for young researchers, grant MK-266.2012.1.}}

\author{F. M. Fominykh \and P. V. Martyugin \and M. V. Volkov}

\date{Institute of Mathematics and Computer Science\\ Ural Federal University, 620000 Ekaterinburg, Russia}

\maketitle

\begin{abstract}
Two topics are presented: synchronization games and synchronization costs. In a synchronization game on a deterministic finite automaton,
there are two players, Alice and Bob, whose moves alternate. Alice wants to synchronize the given automaton, while Bob aims to make her
task as hard as possible. We answer a few natural questions related to such games. Speaking about synchronization costs, we consider
deterministic automata in which each transition has a certain price. The problem is whether or not a given automaton can be synchronized
within a given budget. We determine the complexity of this problem.
\end{abstract}

\section{Introduction and overview}
\label{intro}

A complete deterministic finite automaton (DFA) $\mathrsfs{A}=(Q,\Sigma)$ (here and below $Q$ stands for the state set and $\Sigma$ for the
input alphabet) is called \emph{synchronizing} if there exists a word $w\in\Sigma^*$ whose action brings $\mathrsfs{A}$ to one particular
state no matter at which state $w$ is applied: $q{\cdot}w=q'{\cdot}w$ for all $q,q'\in Q$. Any word $w$ with this property is said to be a
\emph{reset} word for the automaton.

Synchronizing automata serve as transparent and natural models of error-resistant systems in many applications (coding theory, robotics,
testing of reactive systems) and reveal interesting connections with symbolic dynamics, substitution systems and other parts of
mathematics. The literature on \sa\ and their applications is rapidly growing so that even the most recent
surveys~\cite{Sandberg:2005,Volkov:2008} are becoming obsolete. A majority of research in the area focuses on the so-called \v{C}ern\'{y}
conjecture but the theory of \sa\ also offers many other interesting questions. In the present paper we introduce two new directions of the
theory and obtain some initial results in these directions.

Section~\ref{play} concerns with synchronization games on DFAs. In such a game on a~DFA $\mathrsfs{A}$, there are two players, Alice
(Synchronizer) and Bob (Desynchronizer), whose moves alternate. Alice who plays first wants to synchronize $\mathrsfs{A}$, while Bob aims
to prevent synchronization or, if synchronization is unavoidable, to delay it as long as possible. Provided that both players play
optimally, the outcome of such a game depends only on the underlying automaton so studying synchronization games may be considered as a way
to study \sa. The most natural questions here are the following. Given a DFA $\mathrsfs{A}$, how to decide who wins in the synchronization
game on $\mathrsfs{A}$? If Alice wins, how many moves may she need in the worst case, in particular, is there a polynomial of~$n$ that
bounds from above the number of moves in any game on a DFA with $n$ states for which Alice has a winning strategy? How difficult is it to
predict whether or not Alice can win after a certain number of moves? It turns out that these questions can be answered by applying more or
less standard techniques. This may be a bit disappointing but as a byproduct, we reveal a somewhat unexpected relation between
synchronization games and a version of the \v{C}ern\'{y} conjecture.

In Section~\ref{pay} we consider weighted automata. A \emph{deterministic weighted automaton} (DWA) is a DFA $\mathrsfs{A}=(Q,\Sigma)$
endowed with a function $\gamma:Q\times\Sigma\to\mathbb{Z}_+$ where $\mathbb{Z}_+$ stands for the set of all positive integers. In other
words, each transition of a DWA has a certain price being a positive integer. Then every computation performed by $\mathrsfs{A}$ also gets
a certain cost, namely, the sum of the costs of the transitions involved. If a DWA happens to be synchronizing and $w\in\Sigma^*$ is its
\sw, then one can assign to $w$ a cost measured, say, by the maximum among all costs of applying the word $w$ at a state in $Q$. While in
the non-weighted case one is usually interested in minimizing synchronization time, that is, the length of \sws, in the weighted case it is
quite natural to minimize synchronization costs. A basic problem here is to determine, whether or not a given DWA can be synchronized
within a given budget $B\in\mathbb{Z}_+$, in other words, whether or not $\mathrsfs{A}$ admits a \sw\ whose cost does not exceed $B$. We
demonstrate that this problem is PSPACE-complete.

Besides initial questions discussed in this paper, each of the two outlined research directions leads to several intriguing open problems.
We present and  briefly discuss two such problems in Section~\ref{problems}.

The paper has grown from the extended abstract \cite{Fominykh&Volkov:2012} by the first and the third authors. The second author has solved
one of the problems left open in~\cite{Fominykh&Volkov:2012} and his solution has been incorporated in the present paper.

\section{Playing for synchronization}
\label{play}

The idea to consider synchronization as a game has independently arisen in~\cite{Ananichev&Volkov&Zaks:2007}
and~\cite{Blass&Gurevich&Nachmanson&Veanes:2006}. In~\cite{Ananichev&Volkov&Zaks:2007} a one-player game has been used to prove a lower
bound on the minimum length of \sws\ for a certain series of `slowly' \sa. In~\cite{Blass&Gurevich&Nachmanson&Veanes:2006} a specific
synchronization process arising in software testing has been analyzed in terms of a two-player game. The game that we consider here
basically follows the model of~\cite{Ananichev&Volkov&Zaks:2007} but is a two-player game as
in~\cite{Blass&Gurevich&Nachmanson&Veanes:2006}. A further game-theoretic setting related to synchronization has been recently suggested
in~\cite{Jungers:2012}.

Now we describe the rules of our synchronization game. It is played by two players, Alice and Bob say, on an arbitrary but fixed DFA
$\mathrsfs{A}=(Q,\Sigma)$. In the initial position each state in $Q$ holds a coin but, as the game progresses, some coins may be removed.
The game is won by Alice when all but one coins are removed. Bob wins if he can keep at least two coins unremoved indefinitely long.

Alice moves first, then players alternate moves. The player whose turn it is to move proceeds by selecting a letter $a\in\Sigma$. Then, for
each state $q\in Q$ that held a coin before the move, the coin advances to the state $q{\cdot}a$. (In the standard graphical representation
of $\mathrsfs{A}$ as the labelled digraph with $Q$ as the vertex set and the labelled edges of the form $q\xrightarrow{a}q{\cdot}a$, one
can visualize the move as follows: all coins simultaneously slide along the edges labelled $a$.) If after this several coins happen to
arrive at the same state, all of them but one are removed so that when the move is completed, each state holds at most one coin.

\begin{figure}[tb]
\begin{center}
\unitlength=.78mm
\begin{picture}(131,90)(5,-50)\nullfont
\node[Nframe=n,Nfill=n](b0)(43.0,30.0){}
\node[Nframe=n,Nfill=n](b1)(33.0,0.0){}
\drawedge[ELside=r,linewidth=.7,AHdist=2.41,AHangle=20,AHLength=2.5,AHlength=2.41](b0,b1){Move $b$}
\node[Nframe=n,Nfill=n](b0)(103.0,30.0){}
\node[Nframe=n,Nfill=n](b1)(113.0,0.0){}
\drawedge[linewidth=.7,AHdist=2.41,AHangle=20, AHLength=2.5,AHlength=2.41](b0,b1){Move $a$}
\drawcircle[fillgray=0.4](71.0,50.0,3)
\node(n10)(71.0,50.0){}
\drawcircle(47.0,38.0,3)
\node(n11)(47.0,38.0){}
\drawcircle[fillgray=0.9](95.0,38.0,3)
\node(n12)(95.0,38.0){}
\drawcircle[fillgray=0](59.0,18.0,3)
\node(n13)(59.0,18.0){}
\node(n14)(87.0,18.0){}
\drawedge[ELdist=1.77](n10,n12){$a,b$}
\drawedge[ELdist=2.33](n12,n14){$b$}
\drawedge[ELdist=1.59](n14,n13){$b$}
\drawedge[ELdist=2.24](n13,n11){$b$}
\drawedge[ELdist=1.83](n11,n10){$b$}
\drawloop[loopangle=-215.88](n11){$a$}
\drawloop[loopangle=215.88](n13){$a$}
\drawloop[ELdist=1.71,loopangle=-33.69](n14){$a$}
\drawloop[ELdist=2.08,loopangle=37.3](n12){$a$}
\drawcircle(30.0,-8.0,3)
\node(n0)(30.0,-8.0){}
\drawcircle[fillgray=0](6.0,-20.0,3)
\node(n1)(6.0,-20.0){}
\drawcircle[fillgray=0.4](54.0,-20.0,3)
\node(n2)(54.0,-20.0){}
\node(n3)(18.0,-40.0){}
\drawcircle[fillgray=0.9](46.0,-40.0,3)
\node(n4)(46.0,-40.0){}
\drawedge[ELdist=1.77](n0,n2){$a,b$}
\drawedge[ELdist=2.33](n2,n4){$b$}
\drawedge[ELdist=1.59](n4,n3){$b$}
\drawedge[ELdist=2.24](n3,n1){$b$}
\drawedge[ELdist=1.83](n1,n0){$b$}
\drawloop[loopangle=-215.88](n1){$a$}
\drawloop[loopangle=215.88](n3){$a$}
\drawloop[ELdist=1.71,loopangle=-33.69](n4){$a$} \drawloop[ELdist=2.08,loopangle=37.3](n2){$a$}
\node(n5)(112,-8.0){}
\drawcircle(88.0,-20.0,3)
\node(n6)(88,-20.0){}
\drawcircle[fillgray=0.9](138.0,-20.0,3)
\node(n7)(138,-20.0){}
\drawcircle[fillgray=0](100.0,-40.0,3)
\node(n8)(100,-40.0){} \node(n9)(128,-40.0){}
\drawedge[ELdist=2.31](n5,n7){$a,b$}
\drawedge[ELdist=2.21](n7,n9){$b$}
\drawedge[ELdist=1.85](n9,n8){$b$}
\drawedge[ELdist=1.86](n8,n6){$b$}
\drawedge[ELdist=1.83](n6,n5){$b$}
\drawloop[loopangle=-220.86](n6){$a$}
\drawloop[loopangle=220.86](n8){$a$}
\drawloop[ELdist=1.55,loopangle=-35.18](n9){$a$}
\drawloop[ELdist=1.84,loopangle=26.57](n7){$a$}
\end{picture}
\caption{Moves in a synchronization game} \label{moves}
\end{center}
\end{figure}
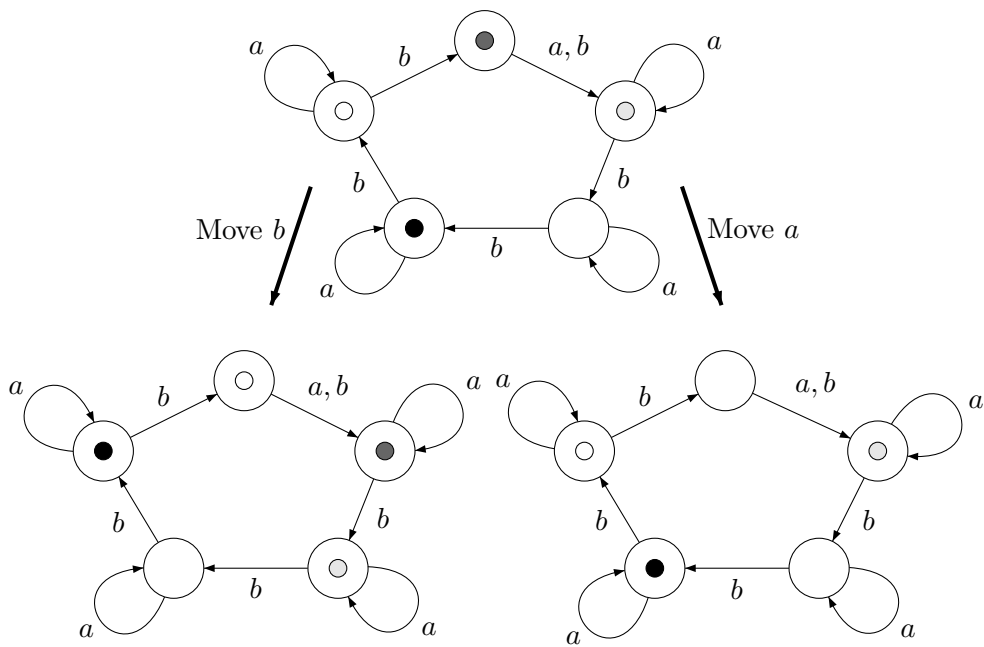

Fig.~\ref{moves} illustrates the rules. Its upper part shows a typical position in a game on a 5-state automaton with 2 input letters $a$
and $b$. The left lower part shows the effect of the move $b$ while the right lower part demonstrates the result of the move $a$. Observe
that in the latter case the dark-gray coin has been removed because it and the light-gray coin had arrived at the same state.

Let $a_1,a_2,\dots,a_k$ with $a_i\in\Sigma$ be a sequence of moves in the synchronization game on $\mathrsfs{A}=(Q,\Sigma)$ and let
$w=a_1a_2\cdots a_k$. It is easy to see that the set of states holding coins after this sequence of moves coincides with the image of $Q$
under the action of the word $w$. Thus, sequences of moves that lead to Alice's win correspond precisely to \sws\ for $\mathrsfs{A}$.
Therefore Bob wins on each DFA which is not synchronizing. Can he win on a \san? Yes, he can: for instance, we show that Bob wins on
automata in the famous \v{C}ern\'{y} series.

\v{C}ern\'{y}~\cite{Cerny:1964} found for each $n>1$ a \san\ $\mathrsfs{C}_n$ with $n$ states and 2 input letters whose shortest reset word
has length $(n-1)^2$. The states of $\mathrsfs{C}_n$ are the residues modulo $n$ and the input letters $a$ and $b$ act as follows:
$$m{\cdot}a=
 \begin{cases}
  1 & \text{for $m = 0$}, \\
  m & \text{for $1\le m<n$};
  \end{cases}
\qquad m{\cdot}b=m+1\!\!\pmod{n}.$$ The automaton is shown in Fig.~\ref{fig:cerny-n}.
\begin{figure}[th]
\begin{center}
\unitlength .8mm
\begin{picture}(110,35)(0,-10)
\node(A)(15,-5){\small$n{-}2$} \node(B)(30,15){\small$n{-}1$} \node(C)(55,22){\small{0}} \node(D)(80,15){\small{1}}
\node(E)(95,-5){\small{2}} \drawloop[loopangle=150](A){$a$} \drawloop[ELpos=40,loopangle=120](B){$a$}
\drawloop[ELpos=60,loopangle=60](D){$a$} \drawloop[loopangle=30](E){$a$} \drawedge(A,B){$b$} \drawedge(B,C){$b$}
\drawedge[curvedepth=3](C,D){$a$} \drawedge[curvedepth=-3,ELside=r](C,D){$b$} \drawedge(D,E){$b$} \put(13,-14){\dots} \put(92,-14){\dots}
\end{picture}
 \caption{The \v{C}ern\'{y} automaton $\mathrsfs{C}_n$}  \label{fig:cerny-n}
\end{center}
\end{figure}
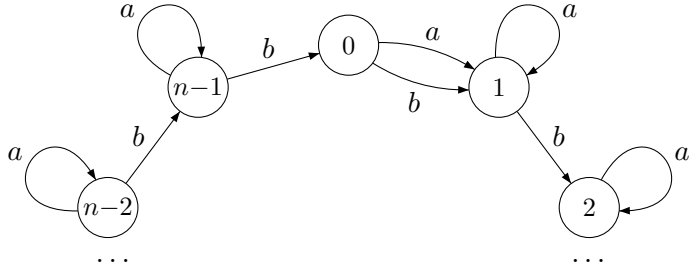

\begin{example}
\label{ex:cerny-n} For each $n>3$, Bob has a winning strategy in the synchronization game on $\mathrsfs{C}_n$.
\end{example}

\begin{proof}
Observe that in $\mathrsfs{C}_n$, the only state where two coins can meet is the state $1$; moreover, this can happen only provided the
move $a$ has been played and before the move the states $0$ and $1$ both held a coin. We may assume for certainty that in this situation it
is the coin arriving from $0$ that is removed after the move.

Under this convention, the winning strategy for Bob is as follows. Bob only has to trace the coins that cover the states $n-1$ and $1$ in
the initial position. For his moves he must always select the letter $a$ except two cases: when the chosen coins cover either $n-2$ and $0$
or $0$ and $2$ in which cases Bob must select~$b$. This way Bob can always keep the coins two steps apart from each other thus preventing
them of being removed.
\end{proof}

On the other hand, it is easy to find DFAs on which Alice has a winning strategy. For instance, a DFA $\mathrsfs{A}$ is called
\emph{definite} in~\cite{Perles&Rabin&Shamir:1963} if there exists an $n>0$ such that every input word of length at least $n$ is a \sw\ for
$\mathrsfs{A}$. Clearly, on each definite automaton, Alice always wins by selecting her moves at random.

The rules of our game readily guarantee that, given a DFA $\mathrsfs{A}=(Q,\Sigma)$, one of the players must have a winning strategy in the
synchronization game on $\mathrsfs{A}$. If Alice has a winning strategy, consider a shortest winning sequence of her moves. Then it is
clear that each move in this sequence creates a position that could not have appeared after an earlier move. However, the number of
possible positions of the game does not exceed $2^{|Q|}-1$ since each position is specified by the subset of states that currently hold
coins. Therefore, if Alice has a winning strategy, she should be able to win after less than $2^{|Q|}$ moves. Thus, one can decide which
player has a winning strategy in the game on $\mathrsfs{A}$ by an exhaustive search through all $|\Sigma|^{2^{|Q|+1}}$ words of length
$2^{|Q|+1}$ over $\Sigma$ in which letters in the odd positions are Alice's moves and ones in the even positions are Bob's replies. Of
course, this brute force procedure is extremely inefficient and it is natural to ask whether an efficient---say, polynomial in the number
of states---algorithm exists. A positive answer can be deduced from the next observation.

\begin{lemma}
\label{localization} Alice has a winning strategy in the synchronization game on a DFA if and only if she has a winning strategy in every
position in which only two states of the DFA hold coins.
\end{lemma}

\begin{proof}
If Alice has no winning strategy for a position $P$ with two coins, $C$ and $C'$ say, then Bob has a winning strategy for $P$. If Bob plays
in the initial position according to this strategy, that is, selects his moves only on the basis of the location of the coins $C$ and $C'$,
as if there were no other coins, the two coins persist forever so that Alice loses the game. (Here we assume that whenever one of the coins
$C$ and $C'$ meets some third coin on some state in the course of the game, then it is this third coin that gets removed.)

Conversely, if Alice can win in every position in which only two states hold coins, she can use the following strategy. In the initial
position she chooses a pair of coins, $C$ and $C'$ say, and plays as if there were no other coins, that is, she applies her winning
strategy for the position in which $C$ and $C'$ cover the same states as they do in the initial position and all other coins are removed.
This brings the game to a position in which either $C$ or $C'$ is removed. Then Alice chooses another pair of coins and again plays as if
these were the only coins, and so on. Since at least one coin is removed in each round, Alice eventually wins.
\end{proof}

Observe that Lemma~\ref{localization} implies a cubic (in the number $n$ of states of the underlying DFA) upper bound on the number of
moves in any game that Alice wins. Indeed, suppose she uses the strategy just described and works with a pair of coins $C$ and $C'$. Let
$q_i$ and $q'_i$ be the states holding the coins $C$ and $C'$ after the $i^{th}$ move of Alice. Then if Alice plays optimally, we must have
$\{q_i,q'_i\}\ne\{q_{i+j},q'_{i+j}\}$ whenever $j>0$. Indeed, the equality $\{q_i,q'_i\}=\{q_{i+j},q'_{i+j}\}$ means that wherever Alice
moves $C$ and $C'$ by her $(i+1)^{th}$, \dots, $(i+j-1)^{th}$ moves, Bob can force Alice to return the coins by her $(i+j)^{th}$ move to
the same states that the coins occupied after her $i^{th}$ move. Then Bob can force Alice to return $C$ and $C'$ to the same states also by
her $(i+2j)^{th}$, $(i+3j)^{th}$, \dots\ moves, whence none of the two coins can ever be removed, a contradiction.

Hence the number of Alice's moves in any round in which she works with any fixed pair of coins does not exceed $\binom{n}2$. Moreover, in
every \san\ there exist states $q$ and $q'$ such that $q{\cdot}a=q'{\cdot}a$ for some letter $a$. Therefore Alice can remove one coin by
her first move. After that she needs at most $n-2$ rounds to remove $n-2$ of the remaining $n-1$ coins. We thus obtain:

\begin{corollary}
\label{cubic} If Alice has a winning strategy in the synchronization game on a DFA with $n$ states, she can win in at most
$\binom{n}2(n-2)+1$ moves.
\end{corollary}

Now we return to the decidability question.
\begin{theorem}
\label{algorithm} Let $\mathrsfs{A}=(Q,\Sigma)$ be a DFA with $|Q|=n$ and $|\Sigma|=k$. There exists an algorithm that in $O(n^2k)$ time
decides who has a winning strategy in the synchronization game on $\mathrsfs{A}$.
\end{theorem}

\begin{proof}
We describe the algorithm rather informally. First we construct a new DFA $\mathrsfs{P}=(P\times\{0,1\}\cup\{s\},\Sigma)$ where $P$ is the
set of all positions with two coins (each such position is specified by a couple of states holding coins) and $s$ is an extra state. The
action of the letters is defined as follows: all letters fix $s$ and if $p\in P$ is the position in which two states $q,q'\in Q$ hold
coins, $x\in\{0,1\}$, and $a\in\Sigma$, then
$$(p,x){\cdot}a=
  \begin{cases}
    (p',1-x) & \text{if $q{\cdot}a\ne q'{\cdot}a$},\\
    s & \text{otherwise},
  \end{cases}$$
where $p'$ is the position in which $q{\cdot}a$ and $q'{\cdot}a$ hold coins. Thus, the automaton $\mathrsfs{P}$ encodes `transcripts' of
all games starting in positions in $P$; the extra bit $x$ controls whose turn it is to move: Alice moves if $x=0$ and Bob moves if $x=1$.
Clearly, $\mathrsfs{P}$ has $n^2-n+1$ states and $k(n^2-n+1)$ edges (transitions).

We mark the state $s$ and then recursively propagate the marking to $P\times\{0,1\}$: a state of the form $(p,0)$ is marked if and only if
there is an $a\in\Sigma$ such that $(p,0){\cdot}a$ is marked and a state of the form $(p,1)$ is marked if and only if for all $a\in\Sigma$
the states $(p,1){\cdot}a$ are marked. Clearly, the marking can be done by a breadth-first search in the underlying digraph of
$\mathrsfs{P}$ with all edges reversed. The well known time estimate for breadth-first search in a graph with $v$ vertices and $e$ edges is
$O(v+e)$, see, e.g., Section~22.2 in~\cite{Cormen&Leiserson&Rivest&Stein:2001}, whence we conclude that the marking can be completed in
$O(n^2k)$ time. It follows from the construction of $\mathrsfs{P}$ and from the marking rules that Alice can win in the game starting at a
position $p\in P$ if and only if the state $(p,0)$ is marked. This and Lemma~\ref{localization} readily imply that Alice has a winning
strategy in the game on $\mathrsfs{A}$ if and only if all states of the form $(p,0)$ get marked (or, equivalently, all states of
$\mathrsfs{P}$ get marked).
\end{proof}

In contrast, we show that it is rather hard to decide whether or not Alice can win after a certain number of moves. Namely, consider the
following decision problem:

\smallskip

\noindent\textsc{Short SynchroGame:} \emph{given a DFA $\mathrsfs{A}$ and a positive integer $\ell$, is it true that Alice can win the
synchronization game on $\mathrsfs{A}$ in at most $\ell$ moves?}

\begin{theorem}
\label{martyugin} \textsc{Short SynchroGame} is PSPACE-complete.
\end{theorem}

\begin{proof}
We first verify that \textsc{Short SynchroGame} lies in the class PSPACE. Take an arbitrary instance $(\mathrsfs{A}=(Q,\Sigma),\ell)$ of
the problem with $|Q|=n$, $|\Sigma|=k$. If $\ell\ge\binom{n}2(n-2)+1$, then by Corollary~\ref{cubic} Alice can win on $\mathrsfs{A}$ in at
most $\ell$ moves whenever she can win on $\mathrsfs{A}$. Thus, for instances of \textsc{Short SynchroGame} with
$\ell\ge\binom{n}2(n-2)+1$, we can solve the problem even in polynomial time invoking the algorithm from the proof of
Theorem~\ref{algorithm}. Therefore we can restrict the problem to instances with $\ell<\binom{n}2(n-2)+1$. (This explains in particular
that it does not really matter whether $\ell$ is given in binary or in unary.)

For each $P$ being a non-empty subset of $Q$ and each non-negative integer $m\le\ell$, introduce two Boolean variables $A(P,m)$ and $B(P,m)$.
The value of $A(P,m)$ is 1 if and only if Alice wins in at most $m$ moves starting from the position in which only states in $P$ hold
coins. The value of $B(P,m)$ is 1 if and only if Alice wins in at most $m$ moves in every position that can arise after Bob's move in the
position in which only states in $P$ hold coins. Then the answer to the instance $(\mathrsfs{A}=(Q,\Sigma),\ell)$ of \textsc{Short
SynchroGame} is `YES' if and only if $A(Q,\ell)=1$ and there is a straightforward recursion that allows one to calculate the values
of the variables $A(P,m)$ and $B(P,m)$:
\begin{align}
\label{eq:recursion}
A(P,m)&=\bigvee_{a\in\Sigma} B(P{\cdot}a,m-1) &&\text{for $m>0$,}\notag\\
B(P,m)&=\bigand_{a\in\Sigma} A(P{\cdot}a,m) &&\text{for all $m$,}\\
A(P,0)&=1 &&\text{if and only if $P$ is a singleton.}\notag
\end{align}
Here $P{\cdot}a$ stands for the set $\{q{\cdot}a\mid q\in P\}$.

The total number of the variables $A(P,m)$ and $B(P,m)$ is of magnitude $2^{n+1}\ell$ so exponential in the size of the input. Nevertheless
it is easy to unfold the recursion~\eqref{eq:recursion} in polynomial space by a sort of depth-first search. As above, we prefer to
describe our algorithm informally. At its generic step, the algorithm tries to evaluate some $A(P,m)$ or $B(P,m)$. Consider the case of
$A(P,m)$. If $m=0$, the algorithm simply checks whether $P$ is a singleton (in other words, if Alice has won) and sets $A(P,0)=1$ if this
is the case and $A(P,0)=0$ otherwise. If $m>0$, the algorithm verifies if all variables of the form $B(P{\cdot}a,m-1)$ where $a\in\Sigma$
have already been evaluated and if this is the case, evaluates $A(P,m)$ according to~\eqref{eq:recursion}. As soon as the value of $A(P,m)$
has been found, the algorithm stores the value but forgets the set $P$ and the used values of $B(P{\cdot}a,m-1)$ so that the space
previously occupied by these data can be re-used. In the case when some of the variables $B(P{\cdot}a,m-1)$ have not yet been evaluated,
the algorithm postpones evaluation of $A(P,m)$ and tries to evaluate the yet unknown variable $B(P{\cdot}a,m-1)$ with the least $a$
(according to a fixed ordering of the alphabet $\Sigma$). In the same manner the algorithm works when evaluating $B(P,m)$.

Clearly, each set $P$ that appears in the course of the implementation must be of the form $P=Q{\cdot}a_1\cdots a_s$,
where $a_1,\dots,a_s\in\Sigma$ and $2s\le\ell$. When calculating the value of $A(P,\ell-\frac{s}2)$ (if $s$ is even) or
$B(P,\ell-\frac{s+1}2)$ (if $s$ odd), the algorithm only needs to maintain the following data:
\begin{itemize}
\item the sets $Q$, $Q{\cdot}a_1$, $Q{\cdot}a_1a_2$, \dots, $Q{\cdot}a_1\cdots a_{s-1}$ that appear on the way from $Q$ to $P$ and the set $P$ itself;
\item the already known but not yet used values of variables of the form $B(Q{\cdot}b_1,\ell-1)$, $A(Q{\cdot}a_1b_2,\ell-1)$, \dots,
$B(Q{\cdot}a_1\cdots a_{t-1}b_t,\ell-\frac{t+1}2)$ for odd $t\le s$, $A(Q{\cdot}a_1\cdots a_{t-1}b_t,\ell-\frac{t}2)$ for even $t\le s$;
\item pointers to the yet unknown variables of the form $B(Q{\cdot}b_1,\ell-1)$, $A(Q{\cdot}a_1b_2,\ell-1)$, \dots,
$B(Q{\cdot}a_1\cdots a_{t-1}b_t,\ell-\frac{t+1}2)$ for odd $t\le s$,\linebreak $A(Q{\cdot}a_1\cdots a_{t-1}b_t,\ell-\frac{t}2)$ for even
$t\le s$ with the least $b_1,b_2,\dots,b_t$ respectively.
\end{itemize}
At each step one has to store at most $2\ell$ sets of size at most $n$, at most $(k-1)(2\ell-1)$ bits for the already calculated but not
yet used values, and at most $2\ell-1$ pointers to the `next' variables to be evaluated. Thus, a polynomial space suffices.

In order to show that \textsc{Short SynchroGame} is PSPACE-complete, we use a reduction from the well known PSPACE-complete problem
\textsc{QSAT} (Quantified Satisfiability) in its game-theoretic form, see Section~19.1 in~\cite{Papadimitriou:1994}. An instance of
\textsc{QSAT} is a Boolean formula in conjunctive normal form with variables $x_1,\dots,x_n$. Alice and Bob play on such an instance
alternatingly: first Alice assigns a value 0 or 1 to $x_1$, then Bob assigns a value to $x_2$ and so on. Alice wins the game if after all
the variables get some values, the formula becomes true; Bob wins if the formula becomes false.

For the reduction we use Eppstein's construction~\cite{Eppstein:1990}. We reproduce it here for the reader's convenience. Given an
arbitrary instance $\psi$ of \textsc{QSAT} with $n$ variables $x_1,\dots,x_n$ and $m$ clauses $c_1,\dots,c_m$, we construct a DFA
$\mathrsfs{A}(\psi)$ with 2 input letters $a$ and $b$ as follows. The state set $Q$ of $\mathrsfs{A}(\psi)$ consists of $(n+1)m$ states
$q_{i,j}$, $1 \le i \le m$, $1 \le j \le n+1$, and a special state $z$. The transitions are defined by
\begin{align*}
& q_{i,j}{\cdot}a =
\begin{cases}
    z \text{ if the literal $x_j$ occurs in $c_i$},\\
    q_{i,j+1} \text{ otherwise}
\end{cases} && \text{ for $1 \le i \le m$, $1 \le j \le n$;} \\
&q_{i,j}{\cdot}b =
\begin{cases}
    z \text{ if the literal $\neg x_j$ occurs in $c_i$},\\
    q_{i,j+1} \text{ otherwise}
\end{cases} && \text{ for $1 \le i \le m$, $1 \le j \le n$;} \\
&q_{i,n+1}{\cdot}a=q_{i,n+1}{\cdot b}=z{\cdot a}={z\cdot b}=z && \text{ for $1\le i\le m$.}
\end{align*}
Fig.~\ref{fig:A2_example} shows an automaton of the form $\mathrsfs{A}(\psi)$ build for the \textsc{QSAT} instance
\begin{align*}
\psi_0&=\{x_1 \vee x_2\vee x_3,\,\neg x_1 \vee x_2\vee x_3,\, x_1 \vee \neg x_2 \vee x_3,\,\neg x_2 \vee \neg x_3\}.
\end{align*}
If at some state $q\in Q$ in Fig.~\ref{fig:A2_example} there is no outgoing edge labelled $c\in\{a,b\}$, the edge $q\stackrel{c}{\to}z$ is
assumed (those edges are omitted to improve readability).

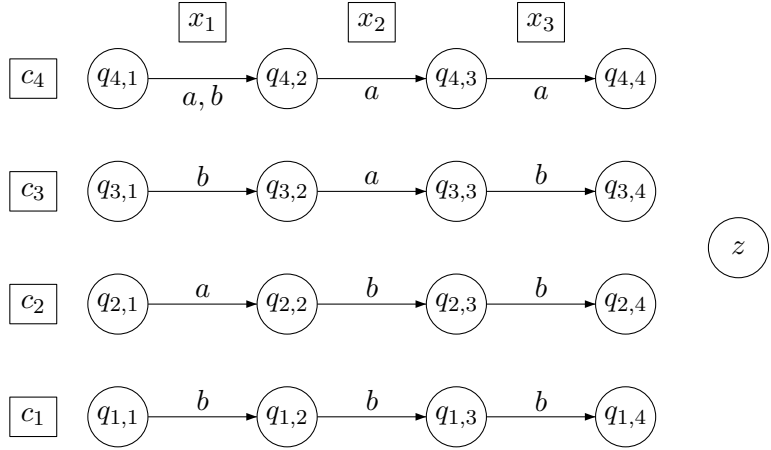
\begin{figure}[tb]
\unitlength=.75mm
\begin{center}
\begin{picture}(120,85)(-100,-10)
\node(n478)(-50,0){$q_{1,2}$} \node(n479)(10,0){$q_{1,4}$} \node(n480)(-80,0){$q_{1,1}$} \node(n481)(-20,0){$q_{1,3}$}
\node(n75)(-50,20){$q_{2,2}$} \node(n32)(-20,20){$q_{2,3}$} \node(n41)(10,20){$q_{2,4}$} \node(n202)(-80,20){$q_{2,1}$}
\node(n42)(10,40){$q_{3,4}$} \node(n172)(-80,40){$q_{3,1}$} \node(n14)(-50,40){$q_{3,2}$} \node(n472)(-20,40){$q_{3,3}$}
\node(n474)(-50,60){$q_{4,2}$} \node(n475)(10,60){$q_{4,4}$} \node(n476)(-80,60){$q_{4,1}$} \node(n477)(-20,60){$q_{4,3}$}

\drawedge(n480,n478){$b$} \drawedge(n478,n481){$b$} \drawedge(n32,n41){$b$} \drawedge(n472,n42){$b$}
\drawedge[ELdist=1.1,ELside=r](n476,n474){$a,b$} \drawedge[ELside=r](n474,n477){$a$} \drawedge[ELside=r](n477,n475){$a$}
\drawedge(n202,n75){$a$} \drawedge(n75,n32){$b$} \drawedge[ELdist=1.1](n172,n14){$b$} \drawedge(n14,n472){$a$}
\drawedge[ELdist=1.1](n481,n479){$b$}

\node[Nw=8.32,Nh=7.0,Nmr=0.0](n1310)(-65,70){$x_1$} \node[Nw=8.32,Nh=7.0,Nmr=0.0](n1316)(-35,70){$x_2$}
\node[Nw=8.32,Nh=7.0,Nmr=0.0](n1318)(-5,70){$x_3$} \node[Nw=8.32,Nh=7.0,Nmr=0.0](n1367)(-95,0){$c_1$}
\node[Nw=8.32,Nh=7.0,Nmr=0.0](n1368)(-95,20){$c_2$} \node[Nw=8.32,Nh=7.0,Nmr=0.0](n1369)(-95,40){$c_3$}
\node[Nw=8.32,Nh=7.0,Nmr=0.0](n1370)(-95,60){$c_4$} \node(n1646)(30,30){$z$}

\end{picture}
\end{center}
\caption{The automaton $\mathrsfs{A}(\psi_0)$} \label{fig:A2_example}
\end{figure}

Observe that Alice wins on $\psi_0$: she may start with letting $x_1=1$ thus ensuring that the first and the third clause become true. Now
if Bob responds by letting $x_2=0$, then the fourth clause becomes true and Alice wins by letting $x_3=1$ which makes also the second
clause be true. If Bob responds by letting $x_2=1$, then the second clause becomes true and Alice wins by letting $x_3=0$.

This winning strategy precisely corresponds to the following winning strategy for Alice in the synchronization game on
$\mathrsfs{A}(\psi_0)$: Alice starts with the move $a$ and if Bob responds with $b$ (respectively $a$), she wins by using $a$ (respectively
$b$).

In general, if Alice has a winning strategy for an instance $\psi$ of \textsc{QSAT} with $n$ variables $x_1,\dots,x_n$ and $m$ clauses
$c_1,\dots,c_m$, she can imitate this strategy in the synchronization game on $\mathrsfs{A}(\psi)$ using the move $a$ whenever she lets
some variable to be true and using the move $a$ whenever she lets some variable to be false. Thus Alice wins the synchronization game in at
most $n$ moves. Conversely, if Alice has a strategy that allows her to win the synchronization game on $\mathrsfs{A}(\psi)$ in at most $n$
moves, she can imitate this strategy in the game on $\psi$ assigning to the current variable values 1 or 0 according to whether her current
move in the game on $\mathrsfs{A}(\psi)$ should be $a$ or $b$. To justify this claim, it suffices to observe that a truth assignment
$\tau:\{x_1,\dots,x_n\}\to\{0,1\}$ enforces $c_i(\tau(x_1),\dots,\tau(x_n))=1$ if and only if the word $a_1a_2\cdots a_n$ of length $n$
defined by
$$a_j=\begin{cases}
a &\text{ if } \tau(x_j)=1,\\
b &\text{ if } \tau(x_j)=0
\end{cases}$$
sends all states $q_{i,1},\dots,q_{i,n+1}$ of the automaton $\mathrsfs{A}(\psi)$ to the state $z$.

Thus, the answer to an instance $\psi$ of \textsc{QSAT} is `YES' if and only if so is the answer to the instance $(\mathrsfs{A}(\psi),n)$
of \textsc{Short SynchroGame} where $n$ is the number of variables of $\psi$. This reduces \textsc{QSAT} to \textsc{Short SynchroGame}.
\end{proof}

Though Corollary~\ref{cubic} and Theorems~\ref{algorithm} and \ref{martyugin}  are worth being registered (as they answer to the most
natural questions related to synchronization games), the reader acquainted with the theory of \sa\ immediately realizes that these results
closely follow some more or less standard patterns. Now we proceed with a more original contribution.

Suppose that Alice has a winning strategy in a synchronization game on an $n$-state DFA. Corollary~\ref{cubic} provides an cubic upper
bound for the number of moves in the game. What about lower bounds? Our next result provides a transparent construction from which we can
extract a quadratic lower bound.

\begin{theorem}
\label{construction} Let $\mathrsfs{A}=(Q,\Sigma)$ be a \san\ with $|Q|=n$, $|\Sigma|\ge2$ and let $\ell$ be the minimum length of \sws\
for $\mathrsfs{A}$. There exists a DFA $\mathrsfs{D}$ with $2n$ states such that Alice wins in the synchronization game on $\mathrsfs{D}$
but needs at least $\ell$ moves for this.
\end{theorem}

\begin{proof}
We fix a letter $b\in\Sigma$ and a state $q_0\in Q$. Now let $\mathrsfs{D}=(Q\times\{0,1\},\Sigma)$ where for each $q\in Q$ the action of
an arbitrary letter $a\in\Sigma$ is defined as follows:
$$(q,0){\cdot}a=(q{\cdot}a,1), \qquad (q,1){\cdot}a=
  \begin{cases}
    (q,0) & \text{if $a=b$}, \\
    (q_0,1) & \text{otherwise}.
  \end{cases}$$
We call $\mathrsfs{D}$ the \emph{duplication} of $\mathrsfs{A}$. Fig.~\ref{fig:duplication} shows the duplication of the \v{C}ern\'{y}
automaton $\mathrsfs{C}_n$ from Fig.~\ref{fig:cerny-n} (with the state $0$ in the role of $q_0$).
\begin{figure}[th]
\begin{center}
    \unitlength=.85mm
       \begin{picture}(140,75)(0,0)
        \rpnode(qn2)(10,10)(20,7){\scriptsize{$(n{-}2,0)$}}
        \rpnode(qn1)(28,52)(20,7){\scriptsize{$(n{-}1,0)$}}
        \rpnode(q0)(70,70)(20,7){\scriptsize{$(0,0)$}}
        \rpnode(q1)(112,52)(20,7){\scriptsize{$(1,0)$}}
        \rpnode(q2)(130,10)(20,7){\scriptsize{$(2,0)$}}
        \rpnode(qn2_)(35,10)(20,7){\scriptsize{$(n{-}2,1)$}}
        \rpnode(qn1_)(45,35)(20,7){\scriptsize{$(n{-}1,1)$}}
        \rpnode(q0_)(70,45)(20,7){\scriptsize{$(0,1)$}}
        \rpnode(q1_)(95,35)(20,7){\scriptsize{$(1,1)$}}
        \rpnode(q2_)(105,10)(20,7){\scriptsize{$(2,1)$}}
        \drawedge[curvedepth=5](qn2,qn1_){$b$}
        \drawedge[curvedepth=5](qn1,q0_){$b$}
        \drawedge[curvedepth=5](q0,q1_){$a,b$}
        \drawedge[curvedepth=5](q1,q2_){$b$}
        \drawedge[curvedepth=5](qn2,qn2_){$a$}
        \drawedge[curvedepth=5](qn1,qn1_){$a$}
        \drawedge[curvedepth=5](q1,q1_){$a$}
        \drawedge[curvedepth=5](q2,q2_){$a$}
        \drawedge[curvedepth=5](qn2_,qn2){$b$}
        \drawedge[curvedepth=5](qn1_,qn1){$b$}
        \drawedge[curvedepth=5](q0_,q0){$b$}
        \drawedge[curvedepth=5](q1_,q1){$b$}
        \drawedge[curvedepth=5](q2_,q2){$b$}
        \drawedge[ELside=r](qn2_,q0_){$a$}
        \drawedge(qn1_,q0_){$a$}
        \drawedge[ELside=r](q1_,q0_){$a$}
        \drawedge(q2_,q0_){$a$}
        \drawloop[loopangle=270](q0_){$a$}
        \put(7,0){$\dots$}
        \put(128,0){$\dots$}
        \put(32,0){$\dots$}
        \put(103,0){$\dots$}
    \end{picture}
    \caption{The duplication of the automaton $\mathrsfs{C}_n$}\label{fig:duplication}
   \end{center}
\end{figure}
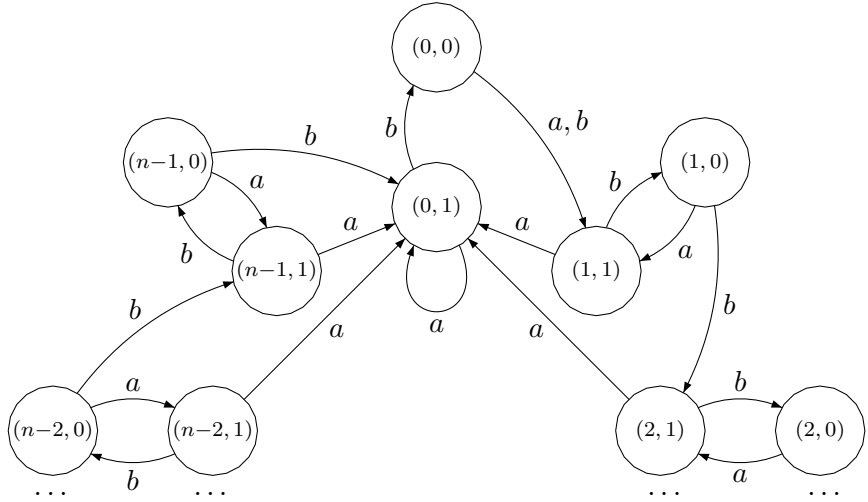

Suppose that Alice opens the game by selecting a letter $a\ne b$. After that only states of the form $(q,1)$ hold coins. Bob must reply
with the move $b$ since he loses immediately otherwise. After that coins cover the states $(q,0)$ with $q\in Q{\cdot}a\cup\{q_0\}$. Now if
Alice spells out a \sw\ for $\mathrsfs{A}$, she wins. Indeed, as soon as Bob selects a letter different from $b$, he loses immediately, and
if he replies with $b$ to all Alice's moves, each pair (Alice's move, Bob's move) has the same effect as applying the letter selected by
Alice in the DFA $\mathrsfs{A}$.

On the other hand, Alice needs at least $\ell$ moves to win if Bob replies with $b$ to each of her moves. Indeed, if Bob plays this way and
a winning sequence of Alice's moves forms a word $w\in\Sigma^*$, then after the last move of the sequence every state $(q,1)$ with $q\in
Q{\cdot}w$ still holds a coin. Thus, for Alice to win, $w$ must be a~\sw\ for $\mathrsfs{A}$, whence the length of $w$ is at least $\ell$.
\end{proof}

We denote by $\mathrsfs{D}_n$ the duplication of the \v{C}ern\'{y} automaton $\mathrsfs{C}_n$. Combining Theorem~\ref{construction} and the
fact that the minimum length of \sws\ for $\mathrsfs{C}_n$ is $(n-1)^2$, we obtain that Alice needs at least $(n-1)^2$ moves to win on
$\mathrsfs{D}_n$. (In fact, the exact number of moves needed is easily seen to be $(n-1)^2+1$.) Thus, we have found a series of $k$-state
DFAs ($k=2n$ is even) on which Alice's win requires a~quadratic in $k$ number of moves. A similar series can be constructed for odd~$k$: we
can just add an extra state to $\mathrsfs{D}_n$ and let both $a$ and $b$ send this added state to the state $(q_0,1) $.

We notice that the duplication of an arbitrary DFA belongs to a very special class of \sa\ as it can be reset by a word of length 2.
A~somewhat unexpected though immediate consequence of Theorem~\ref{construction} is that a~progress in understanding synchronization games
within this specific class may lead to a solution of a major problem in the theory of \sa.

\begin{corollary}
\label{quadartic} If for every $n$-state \san\ with a \sw\ of length~$2$ on which Alice can win, she has a winning strategy with $O(n^2)$
moves, then every $n$-state \san\ has a \sw\ of length $O(n^2)$.
\end{corollary}

Recall that all known results on synchronization of $n$-state DFAs (see~\cite{Pin:1983} and~\cite{Trahtman:2011} for the best bounds)
guarantee only the existence of \sws\ of length $\Omega(n^3)$.

\section{Paying for synchronization}
\label{pay}

Let $\mathrsfs{A}=(Q,\Sigma,\gamma)$ be a DWA, where $\gamma:Q\times\Sigma\to\mathbb{Z}_+$ is a cost function. For $w=a_1\cdots
a_k\in\Sigma^*$ and $q\in Q$, the cost of applying $w$ at $q$ is
$$\gamma(q,w)=\sum_{i=0}^{k-1}\gamma\bigl(q\cdot(a_1\cdots a_i),a_{i+1}).$$
If $\mathrsfs{A}$ is a \san\ and $w$ is its \sw, then the cost of synchronizing $\mathrsfs{A}$ by $w$ is defined as $\gamma(w)=\max_{q\in
Q}{\gamma(q,w)}$. The intuition for this choice of $\gamma(w)$ is as follows: we use $w$ to bring $\mathrsfs{A}$ to a certain state from an
unknown state, and therefore, we have to take the most costly case into account. (Of course, in some situations other definitions of the
cost of synchronization may make sense. For instance, if we treat synchronization in the flavor of Section~\ref{play}, that is, as the
process of moving coins initially placed on all states in $Q$ to a certain state, it is more natural to define the cost of the process as
$\sum_{q\in Q}{\gamma(q,w)}$. The results that follow can be adapted to this setting mutatis mutandis.)

\begin{figure}[bh]
\begin{center}
\unitlength=.95mm
\begin{picture}(20,35)(0,-7)
\node(A)(0,20){0}
\node(B)(20,20){1}
\node(C)(20,0){2}
\node(D)(0,0){3}
\drawedge[curvedepth=4](A,B){$a,{1}$}
\drawedge[curvedepth=-4,ELside=r](A,B){$b,{1}$}
\drawedge(B,C){$b,{1}$} \drawedge(C,D){$b,{1}$}
\drawedge(D,A){$a,{1}$}
\drawloop[loopangle=0](B){$a,{1}$}
\drawloop[loopangle=-0](C){$a,{1}$}
\drawloop[loopangle=-180](D){$b,{16}$}
\end{picture}
\caption{A deterministic weighted automaton}\label{fig:example}
\end{center}
\end{figure}
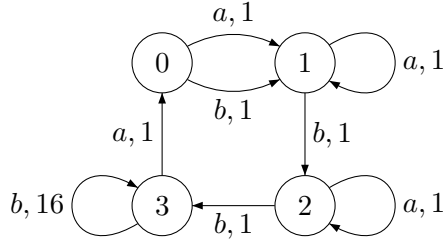

Fig.~\ref{fig:example} shows a DWA (transition costs are included in the labels) and illustrates the difference between two optimization
problems: minimizing synchronization cost and minimizing the length of \sws. The shortest \sw\ for the DWA is $b^3$ but the cost of
synchronizing by $b^3$ is 48. On the other hand, the longer word $a^2baba^2$ manages to avoid the `expensive' loop at the state 3 whence
the cost of synchronizing by $a^2baba^2$ is only 7.

We study in the computational complexity of the following decision problem:

\smallskip

\noindent\textsc{Synchronizing on Budget:} \emph{Given a DWA $\mathrsfs{A}=(Q,\Sigma,\gamma)$ and a positive integer $B$, is it true that
$\mathrsfs{A}$ has a \sw\ $w$ with $\gamma(w)\le B$?}

\smallskip

Here we assume that the values of $\gamma$ and the number $B$ are given in binary. (The unary version of \textsc{Synchronizing on Budget}
can be easily shown to be NP-complete on the basis of the NP-completeness of the problem \textsc{Short Reset
Word}~\cite{Rystsov:1980,Eppstein:1990}: given a DFA $\mathrsfs{A}$ and a positive integer $\ell$, is it true that $\mathrsfs{A}$ has a
reset word of length $\ell$?)

\begin{theorem}
\label{budget} \textsc{Synchronizing on Budget} is PSPACE-complete.
\end{theorem}

\begin{proof}
By Savitch's theorem (see Section~4.3 in~\cite{Papadimitriou:1994}), in order to show that \textsc{Synchronizing on Budget} lies in the
class PSPACE, it suffices to solve this problem in polynomial space by a non-deterministic algorithm. A small difficulty is that for some
instances $(Q,\Sigma,\gamma;B)$ of \textsc{Synchronizing on Budget}, every \sw\ $w$ satisfying $\gamma(w)\le B$ may be exponentially long
in $|Q|$ and so even if our algorithm correctly guesses such a $w$, it would not have enough space to store its guess. To bypass the
difficulty, the algorithm should guess $w$ letter by letter. It guesses the first letter of $w$ (say, $a$), applies $a$ at every state
$q\in Q$ and saves two arrays: $\{q{\cdot}a\}$ and $\{\gamma(q,a)\}$. Each of the arrays clearly requires only polynomial space. Then the
algorithm guesses the second letter of $w$ and updates both arrays, etc. At the end of the guessing steps the algorithm check whether all
entries of the first array are equal (if so, then $w$ is indeed a \sw\ for $(Q,\Sigma)$) and whether the maximum number in the second array
is less than or equal to $B$ (if so, then synchronization is indeed achieved within the budget $B$).

To show that \textsc{Synchronizing on Budget} is PSPACE-complete, we use a~reduction from a problem concerning partial automata. A
\emph{partial} finite automaton (PFA) is a pair $\mathrsfs{A}=(Q,\Sigma)$, where $Q$ is the state set and $\Sigma$ is the input alphabet
whose letters act on $Q$ as partial transformations. Such a PFA is said to be \emph{carefully synchronizing} if there exists $w=a_1\cdots
a_\ell$ with $a_1,\dots,a_\ell\in\Sigma$ such that $q{\cdot}a_i$ with $1\le i\le \ell$ is defined for all $q\in Q{\cdot}(a_1\cdots
a_{i-1})$ and $|Q{\cdot}w|=1$. Every word $w$ with these properties is called a \emph{careful \sw} for $\mathrsfs{P}$. Informally,
a~careful \sw\ synchronizes $\mathrsfs{A}$ and manages to avoid any undefined transition.

The second author~\cite{Martyugin:2010} has recently proved that the following problem is PSPACE-complete:

\smallskip

\noindent\textsc{Careful Synchronization:} \emph{Is a given PFA carefully synchronizing?}

\smallskip

It is the problem that we reduce to \textsc{Synchronizing on Budget}. Our reduction relies on a known fact whose proof is included for the
reader's convenience.

\begin{lemma}
\label{estimate} The minimum length of careful \sws\ for carefully synchronizing PFAs with $n$ states does not exceed $2^n-n-1$.
\end{lemma}

\begin{proof}
Given a PFA $\mathrsfs{A}=(Q,\Sigma)$ with $|Q|=n$, consider the set of the non-empty subsets of $Q$ and let each $a\in\Sigma$ act on
$P\subseteq Q$ as follows:
$$P{\cdot}a=\begin{cases}
    \{q{\cdot}a\mid q\in P\} & \text{provided $q{\cdot}a$ is defined for all $q\in P$}, \\
    \text{undefined} & \text{otherwise}.
  \end{cases}$$
We obtain a new PFA $\mathrsfs{P}$, and it is clear that  $w\in\Sigma^*$ is a careful \sw\ for $\mathrsfs{A}$ if and only if $w$ labels a
path in $\mathrsfs{P}$ starting at $Q$ and ending at a singleton. A~path of minimum length does not visit any state of $\mathrsfs{P}$ twice
and stops as soon as it reaches a singleton. Hence the length of the path does not exceed the number of non-empty and non-singleton subsets
of $Q$, that is, $2^n-n-1$.
\end{proof}

Now take an arbitrary instance of \textsc{Careful Synchronization}, that is, a PFA $\mathrsfs{A}=(Q,\Sigma)$. We assign to $\mathrsfs{A}$
an instance of \textsc{Synchronizing on Budget} as follows. First, extend the action of each letter $a\in\Sigma$ to the whole set $Q$
letting
$$q\odot a=
  \begin{cases}
    q{\cdot}a & \text{if $q{\cdot}a$ is defined in $\mathrsfs{A}$}, \\
    q & \text{otherwise}.
  \end{cases}$$
These extended actions give rise to a DFA $\mathrsfs{A}'$ with the same state set $Q$ and input alphabet $\Sigma$. Further, let $|Q|=n$,
and define $\gamma:Q\times\Sigma\to\mathbb{Z}_+$ by the rule:
$$\gamma(q,a)=
  \begin{cases}
    1 & \text{if $q{\cdot}a$ is defined in $\mathrsfs{A}$}, \\
    2^n & \text{otherwise}.
  \end{cases}$$
This makes $\mathrsfs{A}'$ a DWA. The construction is illustrated by Fig.~\ref{fig:transformation}. Finally, let $B=2^n-1$. Observe that
the binary presentations of $B$ and of the values of $\gamma$ are of a linear in $n$ size so that the construction requires only polynomial
time in the size of the PFA $\mathrsfs{A}$.
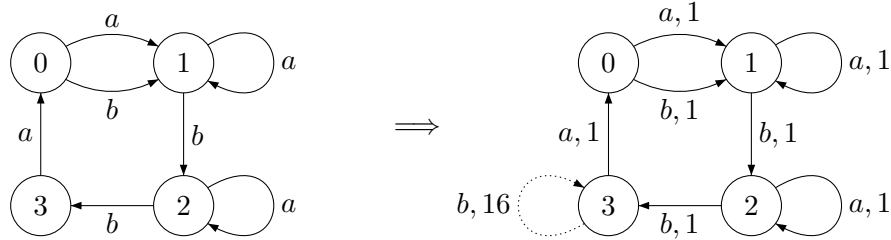
\begin{figure}[ht]
\begin{center}
\unitlength=.95mm
\begin{picture}(30,33)(20,-7)
\node(A)(0,20){0}
\node(B)(20,20){1}
\node(C)(20,0){2}
\node(D)(0,0){3}
\drawedge[curvedepth=4](A,B){$a$}
\drawedge[curvedepth=-4,ELside=r](A,B){$b$}
\drawedge(B,C){$b$}
\drawedge(C,D){$b$}
\drawedge(D,A){$a$}
\drawloop[loopangle=0](B){$a$}
\drawloop[loopangle=-0](C){$a$}
\end{picture}
\qquad
\begin{picture}(30,30)(-20,-7)
\put(-30,10){$\Longrightarrow$}
\node(A)(0,20){0}
\node(B)(20,20){1}
\node(C)(20,0){2}
\node(D)(0,0){3}
\drawedge[curvedepth=4](A,B){$a,{1}$}
\drawedge[curvedepth=-4,ELside=r](A,B){$b,{1}$}
\drawedge(B,C){$b,{1}$}
\drawedge(C,D){$b,{1}$}
\drawedge(D,A){$a,{1}$}
\drawloop[loopangle=0](B){$a,{1}$}
\drawloop[loopangle=-0](C){$a,{1}$}
\drawloop[dash={0.2 0.5}0,loopangle=-180](D){$b,{16}$}
\end{picture}
\caption{Transforming a partial automaton into a weighted automaton}\label{fig:transformation}
\end{center}
\end{figure}

We aim to show that the PFA $\mathrsfs{A}$ is carefully synchronizing if and only if the DWA $\mathrsfs{A}'$ can be synchronized within the
budget $B$. Indeed, if $w$ is a careful \sw\ for $\mathrsfs{A}$, then $w$ can be applied to every state in $\mathrsfs{A}$. This implies
that $w$ labels the same paths in $\mathrsfs{A}'$ as it does in $\mathrsfs{A}$ whence $w$ synchronizes $\mathrsfs{A}'$ and involves only
transitions with cost $1$. Therefore $\gamma(q,w)$ is equal to the length of $w$ for each $q\in Q$ and so is $\gamma(w)=\max_{q\in
Q}{\gamma(q,w)}$. By Lemma~\ref{estimate} $w$ can be chosen to be of length at most $2^n-n-1$, whence $\gamma(w)\le 2^n-n-1<2^n-1=B$.
Conversely, if $w$ is a \sw\ for $\mathrsfs{A}'$ with $\gamma(w)\le B$, then $\gamma(q,w)\le 2^n-1$ for each $q\in Q$, whence no path
labelled $w$ and starting at $q$ involves any transition with cost $2^n$. This means every transition in such a path is induced by a
transition with the same effect defined in $\mathrsfs{A}$. Therefore $w$ can be applied to every state in $\mathrsfs{A}$. Since all paths
labelled $w$ are coterminal in $\mathrsfs{A}'$, they have the same property in $\mathrsfs{A}$ and $w$ is a careful \sw\ for $\mathrsfs{A}$.
\end{proof}

\section{Open problems}
\label{problems}

Due to the space constraint we restrict ourselves to just two interesting problems.

\paragraph{Road Coloring games.} A digraph $G$ in which each vertex has the same out-degree $k$ is called a \emph{digraph of
out-degree} $k$. If we take an alphabet $\Sigma$ of size $k$, then we can label the edges of such $G$ by letters of $\Sigma$ such that the
resulting automaton will be complete and deterministic. Any DFA obtained this way is referred to as a \emph{coloring} of $G$.

The famous Road Coloring Problem asked for necessary and sufficient conditions on a digraph $G$ to admit a synchronizing coloring. The
problem has been recently solved by Trahtman~\cite{Trahtman:2009} and the solution implies that if $G$ has a synchronizing coloring, then
such a coloring can be found in $O(n^2k)$ time where $n$ is the number of vertices and $k$ is the out-degree of $G$,
see~\cite{Beal&Perrin:2008}.

Now consider the following \emph{Rood Coloring game}. Alice and Bob alternately label the edges of a given digraph $G$ of out-degree $k$ by
letters from an alphabet $\Sigma$ of size $k$ (observing the rule that no edges leaving the same vertex may get the same label) until $G$
becomes a DFA. Alice who plays first wins if the resulting DFA is synchronizing, and Bob wins otherwise.

\begin{problem}
Is there an algorithm that, given a digraph $G$ of constant out-degree, decides in polynomial in the size of $G$ time which player has a
winning strategy in the Road Coloring game on $G$?
\end{problem}

Observe that there are digraphs on which Alice wins by making random moves (for instance, the underlying digraphs of the automata in the
\v{C}ern\'{y} series can be shown to have this property); on the other hand, Bob can win on some digraphs admitting synchronizing
colorings, see Fig.~\ref{fig:flip-flop} for a simple example.
\begin{figure}[h]
\begin{center}
\begin{picture}(190,6)(-20,-3)
\unitlength=.5mm \node(B)(0,0){}  \node(D)(40,0){} \drawedge[curvedepth=7](B,D){} \drawedge[curvedepth=7](D,B){}
\drawloop[loopangle=180](B){} \drawloop[loopangle=0](D){} \node(B1)(120,0){}  \node(D1)(160,0){} \drawedge[curvedepth=7](B1,D1){$a$}
\drawedge[curvedepth=7](D1,B1){$b$} \drawloop[loopangle=180](B1){$b$} \drawloop[loopangle=0](D1){$a$}
\end{picture}
\caption{A digraph on which Bob wins the Road Coloring game and its synchronizing coloring}\label{fig:flip-flop}
\end{center}
\end{figure}
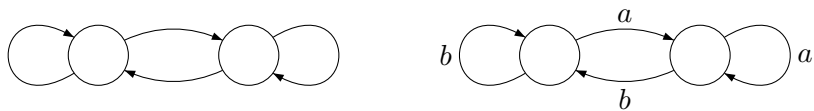

\paragraph{Synchronization games on weighted automata.} As a synthesis of the two topics of this paper, one can consider synchronization
games on DWAs where the aim of Alice is to minimize synchronization costs while Bob aims to prevent synchronization or at least to maximize
synchronization costs. In particular, we suggest to investigate the following problem that can be viewed as a common generalization of
\textsc{Short SynchroGame} and \textsc{Synchronizing on Budget}.

\smallskip

\noindent\textsc{SynchroGame on Budget:} \emph{Given a DWA $\mathrsfs{A}=(Q,\Sigma,\gamma)$ and a positive integer $B$, is it true that
Alice can win the synchronization game on $\mathrsfs{A}$ with a~sequence $w$ of moves satisfying $\gamma(w)\le B$?}

\begin{problem}
Find the computational complexity of \textsc{SynchroGame on Budget}.
\end{problem}

Clearly, the results of the present paper imply that \textsc{SynchroGame on Budget} is PSPACE-hard.

\end{document}